\documentclass[11pt]{llncs} 

\newcommand{\st}{\ensuremath{\overline{s}}}

\newcommand{\Hgame}{G}

\def\smallromani{\renewcommand{\theenumi}{\roman{enumi}}
\renewcommand{\labelenumi}{(\theenumi)}}

\usepackage{amsfonts, stmaryrd}

\usepackage{sgame}

\usepackage{color}
\usepackage{graphics}

\newtheorem{theorem*}{Theorem}

\usepackage{proof}

\newcommand{\opt}{\ensuremath{O}}
\newcommand{\rat}{\ensuremath{rat}}
\newcommand{\interp}[1]{\ensuremath{\llbracket #1 \rrbracket}}
\newcommand{\interpb}[2]{\ensuremath{\llbracket #1 \mid #2 \rrbracket}}
\newcommand{\implies}{\ensuremath{\rightarrow}}
\newcommand{\eol}{ \vspace{4pt} \\ }

\newcommand{\df}[1]{\textup{\bfseries{\emph{#1}}}}

\title{Proof-theoretic Analysis of Rationality for Strategic Games with Arbitrary Strategy Sets}

\author{Jonathan A. Zvesper \inst{1} \and Krzysztof R. Apt \inst{2,3}}
\institute{Oxford University Computing Laboratory, 
Parks Road, Oxford OX1 3QD, UK \and Centre for
    Mathematics and Computer Science (CWI), Science Park 123, 1098~XG
    Amsterdam, the Netherlands
    \and University of Amsterdam, Science Park 904, 1098~XH Amsterdam, the Netherlands}

\begin{document}

\maketitle

\begin{abstract}
  In the context of strategic games, we provide an axiomatic proof
  of the statement
  \begin{description}
\item[(Imp)] Common knowledge of rationality implies that the players will choose
  only strategies that survive the iterated elimination of strictly
  dominated strategies.
  \end{description}
Rationality here means playing only strategies one believes to be best responses.
This involves looking at two formal languages.
One, $\mathcal{L}_O$, is first-order, and is used to formalise optimality conditions, like avoiding strictly dominated strategies, or playing a best response.
The other, $\mathcal{L}_\nu$, is a modal fixpoint language with expressions for optimality, rationality and belief.
Fixpoints are used to form expressions for common belief and for iterated elimination of non-optimal strategies.

\end{abstract}

\section{Introduction}
\label{sec:intro}

There are two main sorts of solution concepts for strategic games:
\emph{equilibrium} concepts
and what might be called ``\emph{effective}'' concepts.
One interpretation of the equilibrium concepts, for example Nash
equilibrium, tacitly presupposes that a game is played repeatedly
(see, e.g. \cite[page 14]{OR_GT}).
Thus the standard condition for Nash equilibrium in terms of the knowledge or beliefs of the players
\cite{AB_Nash} -- the so-called ``epistemic analysis'' of Nash equilibrium --
includes a requirement that players know the other players' strategy choices.
\begin{figure}[hmt]
\begin{center}
\begin{game}{2}{2}
        & $L$   & $R$ \\
$L$     & $1,1$ & $0,0$ \\
$R$     & $0,0$ & $1,1$
\end{game}
~~~
\begin{game}{2}{2}
        & $L$   & $R$ \\
$U$     & $1,1$ & $1,0$ \\
$D$     & $0,0$ & $0,1$
\end{game}
\end{center}
\caption{\label{fig:twogames} Two strategic games}
\end{figure}

Consider the left-hand game in Figure \ref{fig:twogames},
in which each player has two choices $L$ and $R$ and both players get payoff of $1$ if they coordinate, and $0$ otherwise.
Then there are two Nash equilibria\footnote{A Nash equilibrium in a two-player game is a pair $(s_1,s_2)$ of strategies, one for each player such that $s_1$ is a best response to $s_2$ and vice-versa.}: both play $L$ or both play $R$.
But this does not translate by itself into an effective strategy for either player reasoning in isolation,
without some exogenous information.

In contrast, \emph{effective} solution concepts, for example the
iterated elimination of strictly dominated strategies, are compatible
with such a ``one-shot'' interpretation of the game.  Thus the
epistemic analysis of the iterated elimination of strictly dominated
strategies does not require that the players know each other's
strategy choice.  

A strategy $s_i$ is strictly dominated if there is
an alternative strategy $t_i$ such that no matter what the opponent
does, $t_i$ is (strictly) better for $i$ than $s_i$.  Say that a
player is $sd$-rational if he never plays a strategy that he believes
to be strictly dominated.  What the iterated elimination of strictly
dominated strategies does in general require, see \cite{BB_Survey}, is then
that players have \emph{common true belief} that each other is
rational, that is: they are rational, believe that all are rational,
believe that all believe that all are rational, etc.  

In the right-hand game in Figure \ref{fig:twogames}, the column
player, on first looking at her choices $L$ or $R$ is, superficially,
in the same situation as before: choose $L$ and risk the opponent
playing $D$ or choice $R$ and risk the opponent playing $U$.  However,
this time the row player can immediately dismiss playing $D$ on the
grounds that $U$ will \emph{always} be better, no matter what the
column player does.  So if the column player knows (or
believes) this, then he
cannot \emph{rationally} play $R$, and so must play $L$.

In this paper we study the logical form of epistemic characterisation results of this second kind,
so we give formal proof-theoretic principles to justify some given effective or algorithmic process
in terms of common belief of some form of rationality.
We will introduce two formal languages.
One, $\mathcal{L}_O$, is a first-order language, that can be used to define `optimality conditions'.
Avoiding playing a strictly dominated strategy is an example of an `optimality condition'.
Another one is choosing a best response.

However, as observed in \cite{Apt_ManyFaces} for all such notions there are two versions: `local' and `global'.
Notice that in our informal description of when $s_i$ is strictly
dominated by $t_i$ we did not specify \emph{where} $i$ is allowed to
choose alternative strategies from.  In particular, since we are
thinking of an iterated procedure, if $t_i$ has been eliminated
already then it would seem unreasonable to say that $i$ should
consider it.  That intuition yields the \emph{local} definition; the
\emph{global} definition states the opposite: that player $i$ should
always consider his original strategy set from the full game when
looking to see if a strategy is dominated.  

A motivation for looking at global versions of optimality notions is
that they are often mathematically better behaved.  On finite games
the iterations for various local and global versions
coincide
\cite{Apt07c},
but on
infinite games they can differ.  In a nutshell: an optimality condition
$\phi_i$ for player $i$ is \emph{global} if $i$ does not `forget',
during the iterated elimination process, what strategies he has
available in the whole game.  The distinction is clarified in the
respective definitions in $\mathcal{L}_O$.

An optimality condition $\phi$ induces an optimality operator $O_\phi$
on the complete lattice of restrictions (roughly: the subgames) of a
given game.  Eliminating non-$\phi$-optimal strategies can be seen as
the calculation of a fixpoint of the corresponding operator $O_\phi$.
Furthermore, common belief is characterised as a fixpoint (cf.~Note
\ref{note:expr1} below).  Viewed from the appropriate level of
abstraction, in terms of fixpoints of operators, this connection
between common belief of rationality and the iterated elimination of
non-optimal strategies becomes clear.

We define a language $\mathcal{L}_\nu$ that describes things from this higher level of abstraction.
Each optimality condition defines a corresponding notion of \emph{rationality}, which means playing a strategy that one \emph{believes} to be $\phi$-optimal.
$\mathcal{L}_\nu$ is a modal fixpoint language with modalities for \emph{belief} and \emph{optimality},
and so can express connections between optimality, rationality and (common) belief.

We say that an operator $O$ on an arbitrary lattice $(D,\subseteq)$ is
\df{monotonic} when for all $A,B \in D$, if
$A \subseteq B$ then $O(A) \subseteq O(B)$.
The global versions of relevant optimality operators, in particular of the operators corresponding to the best response
and strict dominance, are monotonic.
This is immediately verifiable in $\mathcal{L}_O$ by observing that
the relevant definition is \emph{positive}.

Our first result is a syntactic proof of the following result, where
$\phi$ is a \emph{monotonic} optimality condition:\footnote{By ``common
  true belief'' we mean a common belief that is correct.  In
  particular, common knowledge entails common true belief.}

\begin{theorem*}
Common true belief of $\phi$-rationality entails all played strategies survive the iterated elimination of non-$\phi$-optimal strategies.
\end{theorem*}

Although this theorem relies on a rule for fixpoint calculi that is
only sound for monotonic operators, the semantics of the language
$\mathcal{L}_\nu$ allows also for arbitrary \df{contracting}
operators, i.e.~such that for all $A$, $O(A) \subseteq A$.  We are
therefore able to look at what more is needed in order to justify the
following statement (cf.~\cite[Proposition 3.10]{BB_Survey}), where
$gbr$-rationality means avoiding avoiding strategies one believes to be
never best responses in the global sense:

\begin{theorem*}
\textup{({\bfseries Imp})}
  Common true belief of $gbr$-rationality implies that the players will choose
  only strategies that survive the iterated elimination of strictly
  dominated strategies.
\end{theorem*}

This theorem connects a global notion of $gbr$-rationality with a local one,
referred to in the iterated elimination operator.
Our language allows for arbitrary contracting operators, and their fixpoints to be formed,
and we exhibit one sound rule connecting the resulting fixpoints with monotonic fixpoints.

Our theorems hold for arbitrary games, and the resulting potentially transfinite iterations of the elimination process.
The syntactic approach clarifies the logical underpinnings of the epistemic analysis.
It shows that the use of transfinite iterations can be naturally captured
in $\mathcal{L}_\nu$, at least when the relevant operators are monotonic,
by a single inference rule that involves greatest fixpoints.

The relevance of monotonicity in the context of epistemic analysis of 
finite strategic games has already been pointed out in \cite{vB_RatDyn},
where the connection is also noted between the iterated elimination of non-optimal strategies
and the calculation of the fixpoint of the corresponding operator.

To our knowledge, although several languages have been suggested for reasoning about strategic games
(e.g.~\cite{dB_PHD}),
none use explicit fixpoints (except, as we mentioned, for some suggestions in \cite{vB_RatDyn}) and none use arbitrary optimality operators.

Therefore they are not appropriate for reasoning at the level of abstraction that we suggest 
when studying the epistemic foundations of these ``effective'' solution concepts.
For example while \cite[Section 13]{dB_PHD} does provide some analysis of the logical form of the argument that common knowledge of one kind of rationality implies not playing strategies that are strictly dominated,
the fixpoint reasoning is done at the meta-level.
What \cite{dB_PHD} provides is a proof schema, that shows how, for any finite game,
and any natural number $n$,
to give a proof that common knowledge of rationality entails not playing strategies that are eliminated in $n$ rounds of elimination of non-optimal strategies.

The more general and elegant reasoning principle is captured by using
fixpoint operators and optimality operators.  Another important
advantage to our approach is that we are not restricted in our
analysis to finite games.  This means in particular that our logical
analysis covers the mixed extension of any finite game.

Our use of transfinite iterations is motivated by the original finding of
\cite{Lipman_CKRat}, where a two-player game is constructed for which the
$\omega_0$ (the first infinite ordinal) and $\omega_{0} + 1$ iterations of
the rationalizability operator of \cite{Bernheim} differ.

\section{Games and the language $\mathcal{L}_O$}

A \df{strategic game} is a tuple $(T_1,\ldots,T_n,<_1,\ldots,<_n)$, where
$\{1, \ldots, n\}$ are the players and each $T_i$ is player $i$'s set of strategies, and $<_i$ is player $i$'s preference relation,
which is a total linear order over the set of \df{strategy profiles} $T = \prod_{i=1}^n T_i$.
Note that we assume arbitrary games, rather than restricting to games in which $T$ is finite.
To depict games it is sometimes easier,
as we did in Figure \ref{fig:twogames},
to write down a number for the players' ``payoffs'', rather than just a preference ordering.
We use some standard notation from game theory,
writing $s_{-i}$ for $(s_1, \ldots s_{i - 1}, s_{i+1}, \ldots s_n)$
and $(s_i, t_{-i})$ for the strategy profile
$(t_1, \ldots t_{i - 1}, s_i, t_{i+1}, \ldots s_n)$,
as well as $S_{-i}$ for $\prod_{j \neq i} S_j$.
A \df{restriction} of the game $(T_1,\ldots,T_n,<_1,\ldots,<_n)$ is a sequence $S = (S_1,\ldots,S_n)$ with $S_i \subseteq T_i$ for all players $i$, i.e.~a (possibly empty) subgame in which the payoff information is left out.

The language we use for specifying optimality conditions is a first-order language,
with variables $V = \{x, y, z, \ldots\}$,
a monadic predicate $C$, a constant $o$ and
a family of $n$ ternary relation
symbols $\cdot \geq_{\cdot}^i \cdot$, where $i \in [1..n]$.
So $\mathcal{L}_O$ is given by the following inductive definition:
\[
 \phi ::= C(a) \mid a \geq_c^i b \mid \neg \phi \mid \phi \land \phi \mid \exists x \phi,
\]
where $i \in [1..n]$ and $\{a,b,c\} \subseteq V \cup \{o\}$.

We use the standard abbreviations $\implies$ and $\lor$,
further abbreviate
$\neg \: a \geq_c^i b$ to $b >_c^i a$,
$\forall x \phi$ to $\neg \exists x \neg \phi$,
$\exists x (C(x) \land \phi)$ to $\exists x \in C \: \phi$,
and $\forall x (C(x) \implies \phi)$ to $\forall x \in C \: \phi$.

An \emph{optimality model} $(\Hgame, \Hgame', s)$ is a triple consisting 
of a strategic game $\Hgame = (T_1, \ldots, T_n, <_1, \ldots, <_n)$,
a restriction $\Hgame'$ of $\Hgame$, and a strategy profile $s \in T$.
$G$ will be used to interpret the predicate $C$, and $s$ will be the interpretation of $o$.
An \emph{assignment} for $(\Hgame,\Hgame',s)$ is a function $\alpha$
assigning a strategy profile in $T$ to each variable, and $s$ to $o$.
The ternary satisfaction relation $\models$
between optimality models, assignments and formulas of $\mathcal{L}_O$ is defined inductively as follows,
where $\alpha$ is an assignment for $(\Hgame,\Hgame',s)$, and $\not \models$ the complement of $\models$:
\[\begin{array}{lll}
(\Hgame,\Hgame',s) \models_\alpha C(x)                      & \Leftrightarrow & \forall i \in \{1, \ldots, n\}, \, (\alpha(x))_i \in \Hgame'_i \\
(\Hgame,\Hgame',s) \models_\alpha x \geq_z^i y              & \Leftrightarrow & (\alpha(x)_i,\alpha(z)_{-i}) \geq_i (\alpha(y)_i, \alpha(z)_{-i}) \\
(\Hgame,\Hgame',s) \models_\alpha \neg \phi         & \Leftrightarrow & (\Hgame,\Hgame',s) \not \models_\alpha \phi \\
(\Hgame,\Hgame',s) \models_\alpha \phi_1 \land \phi_2       & \Leftrightarrow & (\Hgame,\Hgame',s) \models_{\alpha} \phi_1 \textup{ and } (\Hgame,\Hgame',s) \models_\alpha \phi_2 \\
(\Hgame,\Hgame',s) \models_\alpha \exists x \phi            & \Leftrightarrow & \textup{there is }\alpha':\,(\Hgame,\Hgame',s) \models_{\alpha'} \phi \textup{ and }\\
                                                                &                         & \forall y \in V \textup{ with } x \neq y, \; \alpha(y) = \alpha'(y)
\end{array}\]
If for any assignment $\alpha$ for $\Hgame$ we have $(\Hgame,\Hgame',s) \models_\alpha \phi$ then we write $(\Hgame,\Hgame',s) \models \phi$.
A variable $x$ occurs \df{free} in $\phi$ if it is not under the scope of a quantifier $\exists x$;
a formula is \df{closed} if it has no free variables.

An \df{optimality condition} for player $i$ is a
closed $\mathcal{L}_O$-formula in which all the occurrences of the atomic formulas $a \geq_c^j b$
are with $j$ equal to $i$.
Intuitively, an optimality condition $\phi_i$ for player $i$ is a way of specifying what it means for
$i$'s strategy in $o$ to be an `OK' choice for $i$ given that $i$'s
opponents will play according to $C_{-i}$ and that $i$'s alternatives are $C_i$.

In particular, we are interested in the following optimality conditions:

\begin{itemize}
 \item $lsd_{i} := \forall y \in C \, \exists z \in C \, o \geq_z^i y$,
 \item $gsd_{i} := \forall y \, \exists z \in C \, o \geq_z^i y$,
 \item $gbr_{i} := \exists z \in C \, \forall y \: o \geq_z^i y$.
\end{itemize}

The optimality conditions listed define some fundamental notions from
game theory: $lsd_i$ says that $o_i$ is not \emph{locally} strictly
dominated in the context of $C$; $gsd_i$ says that $o_i$ is not
\emph{globally} strictly dominated in the context of $C$; and $gbr_i$
says that $o_i$ is globally a best response in the context of $C$.

The distinction between local and global properties, studied further in
\cite{Apt_ManyFaces}, is clarified below.  It important for us here
because the global versions, in contrast to the local ones, satisfy a
syntactic property to be defined shortly.

First, as an illustration of the difference between $gbr_i$ and $gsd_i$, consider the game in Figure \ref{fig:brsd}.
\begin{figure}
\begin{center}
\begin{game}{3}{2}
        & $L$   & $R$ \\
$U$     & $2,1$ & $0,0$ \\
$M$     & $0,1$ & $2,0$ \\
$D$     & $1,0$ & $1,2$
\end{game}
\end{center}
\caption{\label{fig:brsd}An illustration of the difference between strict dominance and best response}
\end{figure}
Call that game $H$, with the row player $1$ and the column player $2$.
Then we have
\[(H,(T_1,T_2),(D,R)) \models gsd_1,\]
but
\[(H,(T_1,T_2),(D,R)) \models \neg gbr_1.\]
The local notions are such that when the `context' restriction $C$ consists of a singleton strategy for a player $i$,
then that strategy is locally optimal.  So for example
\[(H,(\{U,M\},\{R\}),(U,R)) \models lsd_2,\]
whereas
\[(H,(\{U,M\},\{R\}),(U,R)) \models \neg gsd_2.\]

We say that an optimality condition $\phi_i$ is \df{positive} when any sub-formula of the form $C(z)$,
with $z$ any variable, occurs under the scope of an even number of negation signs ($\neg$).
Note that both $gbr_i$ and $gsd_i$ are positive, while $lsd_i$ is not.
As we will see in a moment, positive optimality conditions induce monotonic optimality operators,
and monotonicity will be the condition required of optimality operators in Theorem \ref{thm:main}
relating common knowledge of $\phi$-rationality with the iterated elimination of non-$\phi$ strategies.

\section{Optimality operators}
\label{sec:prelim}

Henceforth let $\Hgame = (T_1,\ldots,T_n,<_1,\ldots,<_n)$ be a fixed
strategic game.  Recall that a \emph{restriction} of the game $\Hgame$
is a sequence $S = (S_1,\ldots,S_n)$ with $S_i \subseteq T_i$ for all
players $i$.  We will interpret optimality conditions as
\emph{operators} on the lattice of the restrictions of a
game ordered by component-wise set inclusion:
\[
\mbox{$(S_1, \ldots, S_n) \subseteq (S'_1, \ldots, S'_n)$ iff $S_i \subseteq S'_i$ for all $i \in [1..n]$}.
\]

Given a sequence $\phi$ giving an optimality condition $\phi_i$ for each player $i$, we introduce an \df{optimality operator} $O_\phi$ defined by
\[
O_\phi(S) = \prod_{i=1}^n\{s_i \in S_i \, \mid \, \phi_i(s_i,S). \}
\]

Consider now an operator $O$ on an arbitrary complete lattice $(D, \subseteq)$ with largest element $\top$.
We say that an element $S \in D$ is a \df{fixpoint} of $O$ if $S = O(S)$
and a \df{post-fixpoint} of $O$ if $S \subseteq O(S)$.

We define by 
transfinite induction a sequence of elements $O^{\alpha}$ of $D$, for all ordinals $\alpha$:
\begin{itemize}
  \item $O^{0} := \top$,
  \item $O^{\alpha+1} := O(O^{\alpha})$,
  \item for limit ordinals $\beta$, $O^{\beta} := \bigcap_{\alpha < \beta} O^{\alpha}$.
\end{itemize}
We call the least $\alpha$ such that $O^{\alpha+1} = O^{\alpha}$ the \df{closure ordinal} of $O$
and denote it by $\alpha_O$.  We call then $O^{\alpha_O}$ the \df{outcome of} (iterating) $O$ and write it alternatively as $O^{\infty}$.


Not all operators have fixpoints, but the monotonic and contracting ones (already defined in the introduction) do:
\begin{note} \label{note:con}
Consider an operator $O$ on $(D, \subseteq)$.
\begin{enumerate}\smallromani
\item If $O$ is contracting or monotonic, then it has an outcome, i.e., 
$O^{\infty}$ is well-defined.
\item The operator $\overline{O}$ defined by $\overline{O}(X) := O(X) \cap X$ is contracting.
\item If $O$ is monotonic, then the outcomes of $O$ and $\overline{O}$ coincide.
\end{enumerate}
\end{note}
\begin{proof}
For (i), it is enough to know that for every set $D$ there is an ordinal $\alpha$ such that there is no injective function from $\alpha$ to $D$.
\end{proof}

Note that the operators $O_\phi$ are by definition contracting, and hence all have outcomes.
Furthermore, it is straightforward to verify that if $\phi_i$ is positive for all players $i$, then $O_\phi$ is monotonic.

The following classic result due to \cite{Tarski_Fixpoint}
also forms the basis of the soundness of some part of the proof systems we consider.\footnote{We
use here its `dual' version in which the iterations start at the largest and not at the least element of a complete lattice.}

\medskip
\noindent
\textbf{Tarski's Fixpoint Theorem} 
For every monotonic operator $O$ on $(D, \subseteq)$
\[
O^{\infty} = \nu O = \cup \{S \in D \mid S \subseteq O(S)\},
\]
where $\nu O$ is the largest fixpoint of $O$.
\vspace{2mm}

We shall need the following lemma,
which is crucial in connecting iterations of arbitrary contracting operators with those of monotonic operators.
It also ensures the soundness of one of the proof rules we will introduce.

\begin{lemma} \label{lem:inc}
Consider two operators $O_1$ and $O_2$ on $(D, \subseteq)$ such that
\begin{itemize}
\item for all $S \in D$, $O_1(S) \subseteq O_2(S)$,

\item $O_1$ is monotonic. 

\end{itemize}
Then $O_1^{\infty} \subseteq \overline{O_2}^{\infty}$.
\end{lemma}
\begin{proof}
By Note \ref{note:con}$(i)$ the outcomes of
$O_1$ and $\overline{O_2}$ exist.

We prove now by transfinite induction that for all $\alpha$
\[
\overline{O_1}^{\alpha} \subseteq \overline{O_2}^{\alpha}
\]
from which the claim follows, since by Note \ref{note:con}$(iii)$
we have $O_1^{\infty} = \overline{O_2}^{\infty}$.

By the definition of the iterations we only need to consider the induction
step for a successor ordinal.  So suppose the claim holds for some
$\alpha$. 

The second assumption implies that $\overline{O_1}$ is monotonic.
We have the following string of inclusions and equalities,
where the first inclusion holds by the induction hypothesis and monotonicity of $\overline{O_1}$
and the second one by the first assumption
\[
\overline{O_1}^{\alpha + 1} =   
\overline{O_1}(\overline{O_1}^{\alpha}) \subseteq 
\overline{O_1}(\overline{O_2}^{\alpha}) =
O_1(\overline{O_2}^{\alpha}) \cap \overline{O_2}^{\alpha} \subseteq 
O_2(\overline{O_2}^{\alpha}) \cap \overline{O_2}^{\alpha} = 
\overline{O_2}^{\alpha + 1}.
\]
\end{proof}

\section{Beliefs and the modal fixpoint language $\mathcal{L}_\nu$}

Recall that $\Hgame$ is a game $(T_1, \ldots, T_n, P_1, \ldots, P_n)$.
A \df{belief model} for $\Hgame$ is a tuple $(\Omega, \st_1, \ldots, \st_n, P_1, \ldots, P_n)$,
with $\Omega$ a non-empty set of `states', and for each player $i$, $\st_i: \Omega \rightarrow T_i$ and $P_i: \Omega \rightarrow 2^\Omega$.
The $P_i$'s are \emph{possibility correspondences} cf.~\cite{BB_Survey}.
The idea of a possibility correspondence $P_i$ is that if the actual state is $\omega$ then $P_i(\omega)$ is the set of states that $i$ considers possible: those that $i$ considers might be the actual state.

Subsets of $\Omega$ are called \df{events}.
A player $i$ \emph{believes} an event $E$ if that event holds in every state that $i$ considers possible.
Thus at the state $\omega$, player $i$ believes $E$ iff $P_i(\omega) \subseteq E$.

Given some event $E$ we write $\Hgame_E$ to denote the restriction of $\Hgame$ determined by $E$:
\[
(\Hgame_E)_i = \{s_i \in T_i \mid \exists u \in E: \st_i(u) = s_i\}.
\]

In the rest of this section we present a formal language $\cal{L}_\nu$ that will be interpreted over belief models.
To begin, we consider the simpler language $\cal{L}$, the formulas of which are
defined inductively as follows, where $i \in [1..n]$:

\[
 \psi ::= \rat_{\phi_i} \mid \psi \land \psi \mid \neg \psi \mid \square_i \psi \mid \opt_{\phi_i} \psi,
\]
with $\phi_i$ an optimality condition for player $i$.
We abbreviate the formula $\bigwedge_{i \in [1..n]} \rat_{\phi_i}$ to $\rat_{\phi}$, 
$\bigwedge_{i \in [1..n]} \square_i \psi$ to $\square \psi$ and
$\bigwedge_{i \in [1..n]} \opt_{\phi_i} \psi$ to $\opt_\phi \psi$.

Formulas of $\cal{L}$ are interpreted as events in (i.e.~as subsets of the domain of) belief models.
Given a belief model $(\Omega, \st_1, \ldots, \st_n, P_1, \ldots, P_n)$ for $\Hgame$,
we define the \df{interpretation function} $\interp{\cdot} : {\cal L} \rightarrow {\cal P}(\Omega)$ as follows:
\begin{itemize}
 \item $\interp{rat_{\phi_i}} = \{ \omega \in \Omega \mid \phi_i(\st_i(\omega), \Hgame_{P_i(\omega)})\}$,
 \item $\interp{\phi \land \psi} = \interp{\phi} \cap \interp{\psi}$,
 \item $\interp{\neg \psi} = \Omega - \interp{\psi}$,
 \item $\interp{\square_i \psi} = \{ \omega \in \Omega \mid P_i(\omega) \subseteq \interp{\psi} \}$,
 \item $\interp{\opt_{\phi_i} \psi} = \{ \omega \in \Omega \mid (\Hgame, \Hgame_{\interp{\psi}}, \st_i(\omega)) \models \phi_i \}$.
\end{itemize}

$P_i(\omega)$ gives the set of states that $i$ considers possible at
$\omega$, so $\interp{rat_{\phi_i}}$ is the event that player $i$ is
$\phi_i$-rational, since it means that $i$'s strategy is optimal
according to $\phi_i$ in the context that the player considers it
possible that he is in.  The semantic clause for $\square_i$ was
mentioned at the begin of this section and is familiar from epistemic
logic: $\interp{\square_i\psi}$ is the event that player $i$ believes
the event $\interp{\psi}$.  $\interp{\opt_{\phi_i} \psi}$ is the event
that player $i$'s strategy is optimal according to the optimality
condition $\phi_i$, in the context of the restriction
$\Hgame_{\interp{\psi}}$.

Then clearly $\interp{rat_\phi}$ is the event that every player $i$ is $\phi_i$-rational;
$\interp{\opt_\phi{\psi}}$ is the event that every player's strategy is $\phi_i$-optimal
in the context of the restriction $\Hgame_{\interp{\psi}}$;
and $\interp{\square \psi}$ is the event that every player believes the event \interp{\psi} to hold.

Although $\cal{L}$ can express some connections between
our formal definitions of optimality rationality and beliefs,
it could be made more expressive.
The language could be extended
with, for example, atoms $s_i$ expressing the event that the strategy $s_i$ is chosen.
This choice is made for example in \cite{dB_PHD}, where modal languages for reasoning about games are defined.
The language we introduce is not parametrised by the game, and consequently can unproblematically be used to reason about games with arbitrary strategy sets.

We will use our language to talk about fixpoint notions: common belief and iterated elimination of non-optimal strategies.
Let us therefore explain what is meant by \df{common belief}.
Common belief of an event $E$ is the event that all players believe $E$, all players believe that they believe $E$, all players believe that they believe that\ldots,  and so on.
Formally, we define $\mathcal{CB}(E)$, the event that $E$ is commonly believed, inductively:
\begin{eqnarray*}
\mathcal{B}_1(E)                & = & \{ \omega \in \Omega \mid \forall i \in [1..n], \, P_i (\omega) \subseteq E \} \\
\mathcal{B}_{m+1}(E)    & = & \mathcal{B}_1(\mathcal{B}_m(E)) \\
\mathcal{CB}(E)         & = & \bigcap_{m > 0}\mathcal{B}_m(E)
\end{eqnarray*}
Notice that $\mathcal{B}_1(E)$ is the event that everybody believes
that $E$ (indeed, we have $\mathcal{B}_1 \interp{\psi} =
\interp{\square \psi}$),  $\mathcal{B}_2(E)$ is the event that
everybody believes that everybody believes that $E$, etc.

`Common belief' is called `common knowledge' when for all players $i$
and all states $\omega \in \Omega$, we have $\omega \in P_i(\omega)$.
In such a case the players have never ruled out the current
state, and so it is legitimate to interpret $\square_i
\psi$ as `$i$ knows that $\psi$'.

Both common knowledge and common belief are known to have equivalent characterisations as fixpoints, and we will exploit this below in defining them in the modal fixpoint language which we now specify.

We extend the vocabulary of ${\cal L}$ with a single
set variable denoted by $X$ and the
contracting
fixpoint operator $\nu X$. 
(The corresponding extension of first-order logic by the dual, inflationary fixpoint operator $\mu X$ was first studied in \cite{DGK_Inflate}.)
Modulo one caveat the resulting language ${\cal L}_\nu$ is defined as follows:
\[
 \psi ::= \rat_{\phi_i} \mid (\psi \land \psi) \mid \neg \psi \mid \square_i \psi \mid \opt_{\phi_i} \psi \mid \nu X. \psi
\]

The caveat is the following: 
\begin{itemize}
\item $\phi$ must be \df{$\nu$-free}, which means that
it does not contain any occurrences of the $\nu X$ operator.
\end{itemize}
This restriction is not necessary but simplifies matters and is sufficient for our considerations.

To extend the interpretation function $\interp{\cdot}$ to ${\cal L}_\nu$,
we must keep track of the variable $X$.
Therefore we first extend the function $\interp{\cdot} : {\cal L} \rightarrow {\cal P}(\Omega)$
to a function $\interpb{\cdot}{\cdot}: {\cal L}_\nu \times {\cal P}(\Omega) \rightarrow {\cal P}(\Omega)$
by padding it with a dummy argument.
We give one clause as an example:
\begin{itemize}
 \item $\interpb{\square_i \psi}{E} = \{ \omega \in \Omega \mid P_i(\omega) \subseteq \interpb{\psi}{E} \}.$
\end{itemize}
We use this extra argument in the semantic clause for the variable $X$:
\begin{itemize}
 \item $ \interpb{X}{E} = E$.
\end{itemize}
Those formulas whose semantics we have so far given define operators.
More specifically, for each of them $\interpb{\psi}{\cdot}$ is an operator on the powerset ${\cal P}(\Omega)$ of $\Omega$.
We use this to define the clause for $\nu X$:
\begin{itemize}
 \item $\interpb{\nu X. \psi}{E} = (\interpb{\psi \wedge X}{\cdot})^{\infty}$.
\end{itemize}
When $X$ does not occur free in $\psi$, we have $\interpb{\psi}{E} = \interpb{\psi}{F}$ for any events $E$ and $F$,
so in these cases we can write simply $\interp{\psi}$.
Note that $\interp{\nu X. \psi}$ is well-defined since for all $E$ we have
$\interpb{\psi \wedge X}{E} = \interpb{\psi}{E} \cap \interpb{X}{E} \subseteq E$, so
the operator $\interpb{\psi \wedge X}{\cdot}$ is contracting.

We say that a formula $\psi$ of $\mathcal{L}_\nu$ is
\df{positive in} $X$ when each occurrence of $X$ in $\psi$ is under the scope of an even number of negation signs ($\neg$),
and under the scope of an optimality operator $\opt_{\phi_i}$ only if $\phi_i$ is positive.
\begin{note}
When $\psi$ is positive, the operator
$\interpb{\psi}{\cdot}$ is monotonic.
\end{note}
Then by Tarski's Fixpoint Theorem and Note \ref{note:con}$(iii)$
we can use the following alternative definition of $\interp{\nu X. \psi}$ in terms of post-fixpoints:
\[
\interp{\nu X. \psi} = \bigcup \{ E \subseteq \Omega \mid E \subseteq \interpb{\psi}{E} \}.
\]

Let us mention some properties the language ${\cal L}_\nu$ can express.
First notice that common belief is definable in ${\cal L}_\nu$
using the $\nu X$ operator.
An analogous characterization of common knowledge 
is in \cite[ Section 11.5]{FHMV_RAK}.

\begin{note}
\label{note:expr1}
Let $\psi$ be a formula of $\cal{L}$.
Then $\interp{\nu X.  \square (X \land \psi)}$ is the event that the event $\interp{\psi}$ is common belief.
\end{note}




From now on we abbreviate the formula $\nu X. \square(X \land \psi)$
with $\psi$ a formula of ${\cal L}$ to $\square^*\psi$.
So $\mathcal{L}_\nu$ can define common belief.
Moreover, as the following observation
shows, it can also define the iterated elimination of non-optimal strategies.
\begin{note} \label{note:expr2} 
In the game determined by the event $\interp{\nu X. \opt_{\phi} X}$, every player selects a strategy which survives the iterated elimination of non-$\phi$-optimal strategies.
\end{note}
\begin{proof}
It follows immediately from the following equivalence, which is obtained by unpacking the relevant definitions:
\[
\Hgame_{\interpb{\opt_\phi X \land X}{E}} = O_\phi(\Hgame_E).
\]
\end{proof}

\section{Proof Systems}

Consider the following formula:
\begin{equation}
\label{eqn:syntax2}
 (\rat_{\phi} \land \square^* \rat_{\phi}) \implies \nu X. \opt_{\phi} X.
\end{equation}

By Notes \ref{note:expr1} and \ref{note:expr2}, we see that (\ref{eqn:syntax2}) states that:
true common belief that the players are $\phi$-rational entails that
each player selects a strategy that survives the iterated elimination of non-$\phi$-optimal strategies.

In the rest of this section we will discuss a simple proof system in which
we can derive (\ref{eqn:syntax2}).
We will use an axiom and rule of inference for the fixpoint operator taken from \cite{K_Mu_Ax}
and one axiom for rationality analogous to the one called in \cite{dB_PHD}
an ``implicit definition'' of rationality. We give these in Figure \ref{fig:proof},
where, crucially, $\psi$ \emph{is positive in $X$, and all the $\phi_i$'s are positive}.
We denote here by $\psi[X \mapsto \chi]$
the formula obtained from $\psi$ by substituting each occurrence of
the variable $X$ with the formula $\chi$.
Assuming given some standard proof rules for propositional reasoning,
we add the axioms and rule given in Figure \ref{fig:proof} to obtain the system \textbf{P}.

\begin{figure}[htbp]
  \begin{center}
    \leavevmode
\fbox{
\begin{tabular}{c}
Axiom schemata \eol

\begin{tabular}{ll}

$\rat_{\phi} \implies (\square \chi \implies \opt_{\phi} \chi)$       & $\rat Dis$  \eol

$\nu X. \psi \implies \psi[X \mapsto \nu X. \psi] $     & $\nu Dis$ \eol

\end{tabular} \eol

Rule of inference \eol

$\infer[\nu Ind]{\chi \implies \nu X. \psi}{\chi \implies \psi[X \mapsto \chi]
}$ \eol

\end{tabular}
}
\caption{\label{fig:proof} Proof system \textbf{P}}
  \end{center}
\end{figure}

A formula is a \df{theorem} of a proof system if it is derivable from the axioms and rules of inference.
An ${\cal L}_\nu$-formula $\psi$ is \df{valid} if
for every belief model $(\Omega, \ldots)$ 
for $\Hgame$ we have
$\interp{\psi} = \Omega$.
We now establish the soundness of the proof system $\textbf{P}$, that is, that its theorems are valid.

\begin{lemma}
The proof system \textbf{P} is sound.
\end{lemma}
\begin{proof}
We show the validity of the axiom $\rat Dis$:

Let $(\Omega, \st_1, \ldots, \st_n, P_i, \ldots, P_n)$ be a belief model for $\Hgame$.
We must show that  $\interp{\rat_{\phi} \implies (\square \chi \implies \opt_{\phi} \chi)} = \Omega$.
That is, that for any $\chi$ the inclusion $\interp{\rat_{\phi}} \cap \interp{\square \chi} \subseteq \interp{\opt_{\phi} \chi}$ holds.
So take some $\omega \in \interp{\rat_{\phi}} \cap \interp{\square \chi}$.
Then for every $i \in [1..n]$, $\phi_i(\st_i(\omega), \Hgame_{P_i(\omega)})$,
and $P_i(\omega) \subseteq \interp{\chi}$.
So by monotonicity of $\phi_i$, $\phi_i(\st_i(\omega), \Hgame_{\interp{\chi}})$, i.e.~$\omega \in \interp{\opt_{\phi_i} \chi}$ as required.

The axioms $\nu Dis$ and the rule $\nu Ind$ were introduced in \cite{K_Mu_Ax};
they formalise, respectively, the following two consequences of Tarski's Fixpoint Theorem
concerning a monotonic operator $F$:
\begin{itemize}
\item $\nu F$ is a post-fixpoint of $F$, i.e., $\nu F \subseteq F(\nu F)$ holds,

\item if $Y$ is a post-fixpoint of $F$, i.e., $Y \subseteq F(Y)$, then
$Y \subseteq \nu F$.
\end{itemize}
\end{proof}

Next, we establish the already announced claim.
\begin{theorem}
\label{thm:main}
The formula (\ref{eqn:syntax2}) is a theorem of the proof system \textbf{P}.
\end{theorem}
\begin{proof}
The following formulas 
are instances of the axioms $\rat Dis$ (with $\psi := \square^* \rat_{\phi} \land \rat_{\phi}$) and $\nu Dis$ (with $\psi := \square (X \land \rat_{\phi})$) respectively: \\
\begin{eqnarray}
 \rat_{\phi} & \implies & (\square (\square^* \rat_{\phi} \land \rat_{\phi}) \implies \opt_{\phi}(\square^* \rat_{\phi} \land \rat_{\phi})), \\
 \square^* \rat_{\phi} & \implies & \square((\square^* \rat_{\phi}) \land \rat_{\phi}).
\end{eqnarray}
Putting these two together via some propositional logic, we obtain
\[
((\square^* \rat_{\phi}) \land \rat_{\phi}) \implies \opt_{\phi} ((\square^* \rat_{\phi}) \land \rat_{\phi}),
\]
which is of the right shape to apply the rule $\nu Ind$
(with $\chi := \square^* \rat_{\phi} \land \rat_{\phi}$ and $\psi := \opt_{\phi} X$). We then obtain
\[ (\square^* \rat_{\phi} \land \rat_{\phi}) \implies \nu X. \opt_{\phi} X, \]
which is precisely the formula (\ref{eqn:syntax2}).
\end{proof}

\begin{corollary}
The formula (\ref{eqn:syntax2}) is valid. 
\end{corollary}

It is
interesting to note that no axioms or rules for the modalities
$\square$ or $\opt$ were needed in order to derive
(\ref{eqn:syntax2}), other than those connecting them with rationality.
In particular, no introspection is required on the part of the players,
nor indeed is the $K$ axiom $\square(\varphi \land \psi) \leftrightarrow (\square \varphi \land \square \psi)$ needed.

In the language ${\cal L}_\nu$, the $\rat_{\phi_i}$ are in effect propositional constants.
We might instead define them in terms of the
$\square_i$ and $\opt_{\phi_i}$ modalities but to this end we would need to extend
the language ${\cal L}_\nu$.
One way to do this is to use a quantified modal language, allowing
quantifiers over set variables, so extending $\mathcal{L}_\nu$ by allowing formulas of the form
$\forall X \varphi$.
Such quantified modal logics are studied in \cite{Fine_QML}.
It is straightforward to extend the semantics to this
larger class of formulas:  
\[
\interpb{\forall X \varphi}{E} = \{\omega \in \Omega \mid \forall F \subseteq \Omega, \, \omega \in \interpb{\varphi}{F} \}.
\]
In the resulting language each $\rat_{\phi_i}$
constant is definable by a formula of this second-order language:
\begin{equation} 
  \label{equ:2nd}
\rat_{\phi_i} \equiv \forall X (\square_i X \implies \opt_{\phi_i} X).
\end{equation}
The following observation then shows correctness of this definition.
\begin{note} For all $i \in [1..n]$
the formula (\ref{equ:2nd}) is valid in the semantics sketched.
\end{note}
To complete our proof-theoretic analysis we augment the proof system \textbf{P} with the following proof rule
where we assume that $\chi$ is
positive in $X$, but where $\psi$ is an arbitrary $\nu$-free $\mathcal{L}_\nu$-formula:
\[
\infer[Incl]{\nu X. \chi \implies \nu X. \psi}{\chi \implies \psi}
\]
The soundness of this rule is a direct consequence of Lemma \ref{lem:inc}.

To formalize the statement \textbf{Imp}
we need two optimality conditions, $gbr_{i}$ and $lsd_{i}$.

To link the proof systems for the languages $\mathcal{L}_O$
and $\mathcal{L}_\nu$ we add the following proof rule, where each $\phi_i$ and $\psi_i$ is an optimality condition in $\mathcal{L}_O$, and $\opt_\phi X \implies \opt_\psi X$ is a formula of $\mathcal{L}_\nu$.
\[
\infer[Link]{\opt_\phi X \implies \opt_\psi X}{\phi_i \implies \psi_i, i \in [1..n]}
\]
The soundness of this rule is a direct consequence of the semantics of the formulas $\opt_\phi X$ and  $\opt_\psi X$.

We denote the system obtained from \textbf{P} by adding to it the above two proof rules 
and standard first-order logic rules concerning the formulas in the language
$\mathcal{L}_O$, like
\[
\infer{\forall x \: \exists y \phi}{\exists y \: \forall x \phi}
\]
by \textbf{R}.
We can now formalize the statement \textbf{Imp} as follows:

\begin{equation}
  \label{eqn:imp}
 (\rat_{gbr} \land \square^* \rat_{gbr}) \implies \nu x. \opt_{lsd} x.  
\end{equation}

The following result then shows that this formula can be formally derived in the considered proof system.
\begin{theorem}
\label{thm:imp}
The formula (\ref{eqn:imp}) is a theorem of the proof system \textbf{R}.
\end{theorem}
\begin{proof}
The properties $gbr_i$ are monotonic, so the following implication
is an instance of (\ref{eqn:syntax2}):
\[(\rat_{gbr} \land \square^* \rat_{gbr}) \implies \nu x. \opt_{gbr} x.\]

Further, since the implication $gbr_i \implies lsd_i$ holds,
we get by the \emph{Link} rule
\[\nu x. \opt_{gbr} x \implies \nu x. \opt_{lsd} x,\]
from which (\ref{eqn:imp}) follows.
\end{proof}

\begin{corollary}
The formula (\ref{eqn:imp}) is valid.
\end{corollary}

\section{Summary}

We have studied the logical form of epistemic characterisation results, for arbitrary (including infinite) strategic games, of the form ``common knowledge of $\phi$-rationality entails playing according to the iterated elimination of non-$\phi'$ properties''.
A main contribution of this work is in revealing,
by giving syntactic proofs,
the reasoning principles involved in two cases: firstly when $\phi = \phi'$ (Theorem \ref{thm:main}), and secondly when $\phi$ entails $\phi'$ (Theorem \ref{thm:imp}).
In each case the result holds when $\phi$ is monotonic.
The language $\mathcal{L}_\nu$ that we used to formalise this reasoning is to our knowledge novel in combining optimality operators with fixpoint notions.  Such a combination is natural when studying such characterisation results, since common knowledge and iterated elimination are both fixpoint notions.

The language $\mathcal{L}_\nu$ is parametric in the optimality conditions used by players.
It is therefore built on the top of a first-order
language $\mathcal{L}_O$ used to define syntactically
optimality conditions relevant for our analysis.

\bibliography{zaclima}
\bibliographystyle{splncs}

\end{document}